\documentclass[a4paper,10pt]{article}
\usepackage[utf8]{inputenc}
\usepackage{amsthm}
\usepackage{amsmath}
\usepackage{authblk}
\usepackage{mathtools}
\usepackage[mathscr]{euscript}
\usepackage{cite}
\usepackage{amssymb}
\usepackage{hyperref}
\newtheorem{theorem}{Theorem}

\newtheorem{lemma}{Lemma}
\newtheorem{corollary}{Corollary}

\newtheorem{observation}{Observation}

\newtheorem{remark}{Remark}
\usepackage{filecontents}

\title{Stability of Stochastic Approximations with `Controlled Markov' Noise
and Temporal Difference Learning}
 \author{Arunselvan~Ramaswamy \texttt{arunr@mail.uni-paderborn.de} \thanks{ Dept. of Electrical Engineering and Information Technology,
 Paderborn University,
 Paderborn - 33908, Germany. His position was funded by the German Research Foundation (DFG) - 315248657.} \thanks{This research was conducted when Ramaswamy was a Ph.D. student at Indian Institute of Science.},
       Shalabh~Bhatnagar \texttt{shalabh@iisc.ac.in} \thanks{
     Department of Computer Science and Automation and the Robert Bosch
 Centre for Cyber Physical Systems, Indian Institute of Science, Bangalore - 560012, India.}}
\begin{document}
\maketitle


%

\begin{abstract}
We are interested in understanding stability (almost sure boundedness) of stochastic approximation algorithms (SAs) driven by a `controlled Markov' process. Analyzing this class of algorithms is important, since many reinforcement learning (RL) algorithms can be cast as SAs driven by a `controlled Markov' process. In this paper, we present easily verifiable sufficient conditions for stability and convergence of SAs driven by a `controlled Markov' process. Many RL applications involve continuous state spaces. While our analysis readily ensures stability for such continuous state applications, traditional analyses do not. As compared to literature, our analysis presents a two-fold generalization (a) the Markov process may evolve in a continuous state space and (b) the process need not be ergodic under any given stationary
 policy. Temporal difference learning (TD) is an important policy evaluation method in reinforcement learning. The theory developed herein, is used to analyze generalized $TD(0)$, an important variant of TD. Our theory is also used to analyze a TD formulation of supervised learning for forecasting problems.
\end{abstract}
\section{Introduction}\label{sec:introduction}
Reinforcement learning (RL) algorithms such as Q learning, temporal difference learning and value iteration methods have seen a major resurgence in recent years as model-free, yet simple and effective, solutions to many important problems. RL is used to solve problems in fields ranging from health-care to transportation. As RL becomes ubiquitous in solving critical problems, there is a need to provide ``behavioral guarantees'' for RL. Stochastic approximation algorithms (SAs) are an important class of model-free algorithms, with associated analytical tools, that play an important role in providing such guarantees. The important foundational papers on SAs include \cite{Ljung77}, \cite{Benaim96,Benaim99}, \cite{BenaimHirsch} and \cite{Borkar99}. Recent results in this field include \cite{MOR, TAC2018} and \cite{andrieu}.

SAs with `controlled' Markov noise are an important subclass of algorithms, particularly since many RL algorithms can be cast in this setting. The groundwork for analyzing such algorithms was laid by Benveniste et. al. \cite{Benveniste} and Borkar \cite{BorkarM}. The analysis of \cite{Benveniste} requires the Markov process to evolve in a finite state space and be ergodic. The analysis of \cite{BorkarM} allows for continuous state spaces and the process may be governed by an additional control-valued sequence in addition to the parameter iterates, a setting that we also consider, and be non-ergodic, \textit{i.e.,} it can have multiple stationary distributions. However, \cite{BorkarM} requires that stability (almost sure boundedness of the algorithm) is ensured. Stability is a highly non-trivial assumption as there is no easy way to ensure compliance with this requirement.

Many RL applications involve Markov processes that evolve over a continuous state space. Here, ensuring stability is especially hard. The main contribution of this paper is the development of easily verifiable sufficient conditions for stability and convergence of SAs driven with an iterate dependent Markov process that may depend on another control-valued sequence. \textit{Our analysis presents a two-fold generalization over traditional ones (a) the Markov process may evolve in a continuous state space and (b) the process need not be ergodic under any given stationary policy.} Under our conditions, the algorithm is shown to be stable and it tracks a solution to a limiting differential inclusion (DI), defined in terms of the ergodic occupation measures of the Markov process. Further, the limiting set is internally chain transitive and invariant. Our stability assumptions are particularly interesting, since they can be readily used to ensure stability in many RL applications, and are compatible with traditional convergence analyses.

Temporal difference learning (TD) is an important RL algorithm that is popularly employed in `policy evaluation' problems. $TD(0)$ is an important variant of TD that is effective, yet simple to implement. Our theory is used to provide a complete analysis of generalized $TD(0)$. Previously, Tsitsiklis and Van Roy \cite{Tsitsiklis} have analyzed TD. However, \cite{Tsitsiklis} assumes that the Markov process is ergodic and evolves in a finite state space. Further, the second moments of the single stage rewards are assumed to be bounded. Our analysis eliminates the need to impose such restrictions, see Section~\ref{application} for details.

As yet another application of our theory, we analyze a \textit{TD formulation of supervised learning}, to solve the weather forecasting problem described in \textit{Chapter 11} of Spall \cite{spall2005}. It may be noted that the analyses in \cite{Tsitsiklis} and \cite{andrieu} cannot be applied to analyze the aforementioned algorithm.
\subsection{Notations \& Definitions} \label{definitions}
\noindent
\textit{\textbf{[Upper-semicontinuous map]}} We say that $H$ is upper-semicontinuous,
  if given sequences $\{ x_{n} \}_{n \ge 1}$ (in $\mathbb{R}^{n}$) and 
  $\{ y_{n} \}_{n \ge 1}$ (in $\mathbb{R}^{m}$)  with
  $x_{n} \to x$, $y_{n} \to y$ and $y_{n} \in H(x_{n})$, $n \ge 1$, 
  then $y \in H(x)$.
\\
\textit{\textbf{[Marchaud Map]}} A set-valued map $H: \mathbb{R}^n \to \{subsets\ of\ \mathbb{R}^m$ \} 
is called \textit{Marchaud} if it satisfies
the following properties:
 \textbf{(i)} for each $x$ $\in \mathbb{R}^{n}$, $H(x)$ is convex and compact;
 \textbf{(ii)} \textit{(point-wise boundedness)} for each $x \in \mathbb{R}^{n}$,  
 $\underset{w \in H(x)}{\sup}$ $\lVert w \rVert$
 $< K \left( 1 + \lVert x \rVert \right)$ for some $K > 0$;
 \textbf{(iii)} $H$ is \textit{upper-semicontinuous}. \\
Let $H$ be a Marchaud map on $\mathbb{R}^d$.
The differential inclusion (DI) given by
\begin{equation} \label{di}
\dot{x} \ \in \ H(x)
\end{equation}
is guaranteed to have at least one solution that is absolutely continuous. 
The reader is referred to \cite{Aubin} for more details.
We say that $\textbf{x} \in \sum$ if $\textbf{x}$ 
is an absolutely continuous map that satisfies (\ref{di}).
The \textit{set-valued semiflow}
$\Phi$ associated with (\ref{di}) is defined on $[0, + \infty) \times \mathbb{R}^d$ as: \\
$\Phi_t(x) = \{\textbf{x}(t) \ | \ \textbf{x} \in \sum , \textbf{x}(0) = x \}$. Let
$B \times M \subset [0, + \infty) \times \mathbb{R}^d$ and define
\begin{equation}\nonumber
 \Phi_B(M) = \underset{t\in B,\ x \in M}{\bigcup} \Phi_t (x).
\end{equation}
\\
\textit{\textbf{[$\omega$-limit set]}}
Given $M \subseteq \mathbb{R}^d$, the $\omega$-limit \textit{set} is defined as
$
 \omega_{\Phi}(M) = \bigcap_{t \ge 0} \ \overline{\Phi_{[t, +\infty)}(M)}.
$
\\
\textit{\textbf{[Limit set of a solution]}} The limit set of a solution $\textbf{x}$
with $\textbf{x}(0) = x$ is given by
$L(x) = \bigcap_{t \ge 0} \ \overline{\textbf{x}([t, +\infty))}$.
\\
\textit{\textbf{[Invariant set]}}
$M \subseteq \mathbb{R}^d$ is \textit{invariant} if for every $x \in M$ there exists 
a trajectory, $\textbf{x} \in \sum$, entirely in $M$
with $\textbf{x}(0) = x$, $\textbf{x}(t) \in M$,
for all $t \ge 0$.
\\ \textit{\textbf{[Open and closed neighborhoods of a set]}}
Let $x \in \mathbb{R}^d$ and $A \subseteq \mathbb{R}^d$, then
$d(x, A) : = \inf \{\lVert a- y \rVert \ | \ y \in A\}$. We define the $\delta$-\textit{open neighborhood}
of $A$ by $N^\delta (A) := \{x \ |\ d(x,A) < \delta \}$. The 
$\delta$-\textit{closed neighborhood} of $A$ 
is defined by $\overline{N^\delta} (A) := \{x \ |\ d(x,A) \le \delta \}$.
The open ball of radius $r$ around the origin is represented by $B_r(0)$,
while the closed ball is represented by $\overline{B}_r(0)$.
\\ \textit{\textbf{[Internally chain transitive set]}}
$M \subset \mathbb{R}^{d}$ is said to be
internally chain transitive if $M$ is compact and for every $x, y \in M$,
$\epsilon >0$ and $T > 0$ we have the following: There exists $n$ and $\Phi^{1}, \ldots, \Phi^{n}$ that
are $n$ solutions to the differential inclusion $\dot{x}(t) \in H(x(t))$,
points $x_1(=x), \ldots, x_{n+1} (=y) \in M$
and $n$ real numbers 
$t_{1}, t_{2}, \ldots, t_{n}$ greater than $T$ such that: $\Phi^i_{t_{i}}(x_i) \in N^\epsilon(x_{i+1})$ and
$\Phi^{i}_{[0, t_{i}]}(x_i) \subset M$ for $1 \le i \le n$. The sequence $(x_{1}(=x), \ldots, x_{n+1}(=y))$
is called an $(\epsilon, T)$ chain in $M$ from $x$ to $y$.
\\ \textit{\textbf{[Attracting set, fundamental neighborhood \& attractor]}}
$A \subseteq \mathbb{R}^d$ is \textit{attracting} if it is compact
and there exists a neighborhood $U$ such that for any $\epsilon > 0$,
$\exists \ T(\epsilon) \ge 0$ with $\Phi_{[T(\epsilon), +\infty)}(U) \subset
N^{\epsilon}(A)$. Such a $U$ is called the \textit{fundamental neighborhood} of $A$. 
In addition to being compact if the \textit{attracting set} is also invariant then
it is called an \textit{attractor}.
The \textit{basin
of attraction } of $A$ is given by $B(A) = \{x \ | \ \omega_\Phi(x) \subset A\}$.
\\
\textit{\textbf{[Lyapunov stable]}} The above set $A$ is Lyapunov stable 
if for all $\delta > 0$, $\exists \ \epsilon > 0$ such that
$\Phi_{[0, +\infty)}(N^\epsilon(A)) \subseteq N^\delta(A)$.
\\ \textit{\textbf{[Upper-limit of a sequence of sets, Limsup]}}
Let $\{K_{n}\}_{n \ge 1}$ be a sequence of sets in $\mathbb{R}^{d}$. 
The \textit{upper-limit} of $\{K_{n}\}_{n \ge 1}$ is
given by, 
$Limsup_{n \to \infty} K_n$ $:= \ \{y \ | \ 
\underset{n \to \infty}{\underline{lim}}d(y, K_n)= 0 \}$. \\
We may interpret that the upper-limit
collects its accumulation points.

\section{Assumptions} \label{assumptions}
As stated earlier, we are motivated by the need to analyze RL algorithms. However, we present our analysis for a more general class of algorithms: SAs driven by an iterate-dependent `controlled Markov process'. Later, we shall recast this analysis to understand generalized $TD(0)$ and a TD formulation of supervised learning with delayed rewards. It may be noted that the analysis of TD for supervised learning was only possible due to our consideration of the general class.

Let us begin by describing a $SA$ driven by a `controlled Markov' process, following which we list the assumptions involved. Since we use results from \cite{BorkarM} in our convergence analysis, we relate our assumptions with those of \cite{BorkarM}. We have the following recursion in $\mathbb{R}^d$:
\begin{equation}\label{eq:sre}
x_{n+1} = x_n + a(n) \left[ h(x_n, y_n) + M_{n+1} \right].
\end{equation}

 \textbf{(A1)(i)} $h: \mathbb{R}^d \times S \to \mathbb{R}^d$ is a jointly continuous map, and
 $S$ a compact metric space. The map $h$ is
 Lipschitz continuous in the first component, with Lipschitz constant $L$, which
 does not change with the second component.\footnote{
 assumption $(1)$ in \textit{Section 2}, \cite{BorkarM}.}
 
  \textbf{(A1)(ii)}$\{y_n\}_{n \ge 0}$ is an $S$-valued `controlled Markov' process controlled by (a) the iterate sequence $\{x_n\}$ and (b) an additional control-valued sequence. \footnote{$S$ is a compact metric space, and hence Polish.}
 
 \textbf{(A2)} $\{M_n\}_{n \ge 1}$ is a square integrable martingale difference sequence (noise). Further, 
 $ E \left[ \lVert M_{n+1}\rVert ^2 \ |\ \mathcal{F}_n \right] \le K(1+ \lVert x_n \rVert ^2), \text{ where } n \ge 0
 $
 and $\mathcal{F}_n := \sigma \left\langle x_m , y_, M_m; m \le n \right\rangle$, $n \ge 0$.
 \footnote{assumption $(2)$ in \textit{Section 2}, \cite{BorkarM}.}
 
 \textbf{(A3)} The step-size sequence $\{a(n)\}_{n \ge 0}$ satisfies the following: 
 it is non-increasing, $a(n) > 0$ for all $n \ge 0$,
 $\sum_{n=0}^\infty a(n) = \infty$ and $\sum_{n=0}^\infty a(n)^2 < \infty$. Without loss of generality
 let $\underset{n \ge 0}{\sup}\ a(n) \le 1$.
 \footnote{assumption $(3)$ in \textit{Section 2}, \cite{BorkarM}.}

Before proceeding, we define a family of rescaled functions as follows. For each $c \ge 1$, define $h_c : \mathbb{R}^d \times S \to \mathbb{R}^d$
 by $h_c(x, y) := h(cx, y)/c$.
Also define, 
$h_\infty: \mathbb{R}^d \times S \to \{\text{subsets of }\mathbb{R}^d\}$ by $h_\infty(x, y) :=$
$Limsup_{c \to \infty} \{h_c(x, y)\}$, where $Limsup$ is the upper-limit of a sequence of sets (see Section~\ref{definitions}). 
Finally, define the set-valued map, $H$, as
$H(x) := \overline{co} \left( \underset{y \in S}{\cup} h_\infty(x, y) \right)$, where $x \in \mathbb{R}^d$. In Lemma~\ref{Hismarchaud} we show that $H$ is Marchaud. Consequently, there exists a solution to the $DI$ $\dot{x}(t) \in H(x(t))$, see \cite{Aubin}
for details.

Below we present our key stability assumptions, $(S1)$ and $(S2)$. They are based on the limiting behavior of the objective function $h$.

 \textbf{(S1)} If $c_n \uparrow \infty$, $y_n \to y$ and 
 $\underset{n \to \infty}{\lim} h_{c_n}(x,y_n)$ = $u$ for some $u \in \mathbb{R}^d$, then
 $u \in h_\infty(x, y)$.
 
 \textbf{(S2)} There exists an attracting set, $\mathcal{A}$, associated with $\dot{x}(t) \in H(x(t))$
 such that $\underset{u \in \mathcal{A}}{\sup}\lVert u \rVert < 1$ and $\overline{B}_1 (0)$
 is a fundamental neighborhood of $\mathcal{A}$.

 It follows from $(S2)$ that we can find $\delta_1, \delta_2, \delta_3$ and $\delta_4$ such that
 $\delta_1 := \underset{u \in \mathcal{A}}{\sup} \ \lVert u \rVert$ and
$\delta_1 < \delta_2 < \delta_3 < \delta_4 < 1$.

\begin{observation}\label{hc_observe}{
$h_c$ is a jointly continuous function that is Lipschitz continuous in the first component. Further,
the Lipschitz constant is independent of the second component and $c$. Without loss of generality we take this Lipschitz constant to be $L$ from $(A1)$.
Since $S$ is compact,
$\lVert h_c(0, \cdot) \rVert _\infty \le \lVert h(0, \cdot) \rVert _\infty \le \overline{M}$ for some $0 < \overline{M} < \infty$. It now follows from
$ \lVert h_c(x,y) - h_c(0,y) \rVert \le L\lVert x \rVert$, that
$\lVert h_c(x,y) \rVert \le \lVert h_c(0,y) \rVert + L \lVert x \rVert$ and
$ \lVert h_c (x, y) \rVert \le  K \left( 1 +  \lVert x \rVert \right)$ are satisfied, where $K := L \vee \overline{M}$.
 Again, without loss of generality this $K$ is from $(A2)$ and does not change with $c$. In other words,
 $\underset{c \ge 1}{\sup} \lVert h_c(x,y) \rVert \le K(1 + \lVert x \rVert)$. Further, it follows from the definition of $h_\infty$ and $H$ that
\begin{equation}\label{Hpwb}
 \underset{u \in H(x)}{\text{sup}}  \lVert u \rVert \le K(1 + \lVert x \rVert). 
\end{equation}
 }
\end{observation}
As stated earlier, we need to show that the set-valued map $H$ is Marchaud. First, we prove a technical lemma that is needed for this purpose.
\begin{lemma}\label{hcross}
 Suppose $x_n \to x$ in $\mathbb{R}^d$, $y_n \to y$ in $S$, $c_n \uparrow \infty$
 and $\underset{c_n \uparrow \infty}{\lim} h_{c_n}(x_n , y_n) = u$. Then $u \in h_\infty(x, y)$.
\end{lemma}
\begin{proof}
 It follows from Observation~\ref{hc_observe}
 that $\lVert h_{c_n} (x, y_n) -  h_{c_n} (x_n, y_n)\rVert \le L \lVert x_n - x \rVert$.
 Since $x_n \to x$ the sequences $\{h_{c_n}(x,y_n)\}_n$ and $\{h_{c_n}(x_n,y_n)\}_n$
 have the same limit as $n \to \infty$ \textit{i.e.,} $u$. It now follows from assumption $(S1)$
 that $u \in h_\infty(x,y)$.
\end{proof}

We claim the following: if $x_n \to x$ in $\mathbb{R}^d$, $\{y_n\} \subset S$
and $c_n \to \infty$ then $d(h_{c_n}(x_n, y_n), H(x)) \to 0$. Suppose this claim were false, then, without loss of generality,
$d(h_{c_n}(x_n, y_n), H(x)) > \epsilon$ for some $\epsilon > 0$, $n \ge 0$. 
Since $S$ is compact, $\exists \{m(n)\} \subseteq \{n\}$ such that $c_{m(n)} \uparrow \infty$,
$\underset{m(n) \to \infty}{\lim} y_{m(n)} = y$ and $h_{c_{m(n)}}(x_{m(n)}, y_{m(n)}) \to u$ for some
$y \in S$ and $u \in \mathbb{R}^d$.
Hence, $x_{m(n)} \to x$, $y_{m(n)} \to y$,
$c_{m(n)} \uparrow \infty$, $h_{c_{m(n)}}(x_{m(n)}, y_{m(n)}) \to u$ and $u \notin h_\infty (x, y) \subseteq H(x)$. This contradicts Lemma~\ref{hcross}.
\begin{lemma}\label{Hismarchaud}
 The set-valued map $H$ is Marchaud.
\end{lemma}
\begin{proof}
Recall that $H(x) = \overline{co}\left( \underset{y \in S}{\cup} h_\infty(x, y) \right)$.
As explained earlier (\textit{cf.} (\ref{Hpwb})),
\[ \underset{u \in H(x)}{\sup} \lVert u \rVert \le K(1 + \lVert x \rVert). \]
Hence $H$ is point-wise bounded. From the definition of $H$ it follows that $H(x)$ is convex
and compact for each $x \in \mathbb{R}^d$.

It is left to show that $H$ is upper semi-continuous.
Let $x_{n} \to x$, $u_{n} \to u$ and $u_{n} \in H(x_n)$, 
$n \ge 1$.
 We need to show that $u \in H(x)$. If this is not true, then
 there exists a linear functional on $\mathbb{R}^{d}$, say $f$, such that
 $\underset{v \in H(x)}{\sup}$ $f(v) \le \alpha - \epsilon$
 and $f(u) \ge \alpha + \epsilon$, for some
 $\alpha \in \mathbb{R}$ and $\epsilon > 0$. 
 Since $u_{n} \to u$, there exists 
 $N$ such that for each $n \ge N$ $f(u_{n}) \ge \alpha + \frac{\epsilon}{2}$, \textit{i.e.},
 $H(x_n) \cap  [f \ge \alpha + \frac{\epsilon}{2}] \neq \phi$, here $[f \ge a]$ is used to denote the set
 $\left\{ x \ |\ f(x) \ge a \right\}$. For the sake of notational convenience let us denote
 $\underset{y \in S}{\cup} h_\infty (x, y)$ by $A(x)$ for all $x \in \mathbb{R}^{d}$.
 We claim that $A(x_{n}) \cap [f \ge \alpha + \frac{\epsilon}{2}] \neq \phi$
 for all $n \ge N$. We shall prove this claim later,
 for now we assume that the claim
 is true and proceed. 
 
 Pick $w_{n} \in A(x_{n}) \cap [f \ge \alpha + \frac{\epsilon}{2}]$
 for each $n \ge N$. 
 Let $w_n \in h_\infty(x_n , y_n)$ for some $y_n \in S$.
 Since $\{w_{n}\}_{n \ge N}$ is norm bounded
 it contains a convergent subsequence, say
 $\{w_{n(k)}\}_{k \ge 1} \subseteq \{w_{n}\}_{n \ge N}$. 
 Let $\underset{k \to \infty}{\lim} w_{n(k)} = w$.
Since $w_{n(k)} \in h_\infty(x_{n(k)}, y_{n(k)})$, 
 $\exists$ $c_{n(k)} \in \mathbb{N}$ such that $\lVert w_{n(k)} - h_{c_{n(k)}}(x_{n(k)}, y_{n(k)}) \rVert
 < \frac{1}{n(k)}$. 
 The sequence
 $\{c_{n(k)}\}_{k \ge 1}$ is chosen such that $c_{n(k+1)} > c_{n(k)}$ for each $k \ge 1$.
 Since $\{y_{n(k)} \}_{k \ge 1}$ is from a compact set, there exists a convergent subsequence. For
 the sake of notational convenience (without loss of generality) we assume that the sequence
 itself has a limit, \textit{i.e.},
 $y_{n(k)} \to y$ for some $y \in S$.
 We have the following: $c_{n(k)} \uparrow \infty$, $x_{n(k)} \to x$, $y_{n(k)} \to y$,
 $w_{n(k)} \to w$ and $w_{n(k)} \in h_{c_{n(k)}}(x_{n(k)}, y_{n(k)})$ for $ k \ge 1$. 
 It follows from Lemma~\ref{hcross} that $w \in h_\infty(x, y)$. Since $w_{n(k)} \to w$
 and $f(w_{n(k)}) \ge \alpha + \frac{\epsilon}{2}$ for each $k \ge 1$, we have that
 $f(w) \ge \alpha + \frac{\epsilon}{2}$. This contradicts
 $\underset{w \in H(x)}{\sup}$ $f(w) \le \alpha - \epsilon$.
 
 It remains to prove that  $A(x_{n}) \cap [f \ge \alpha + \frac{\epsilon}{2}] \neq \phi$
 for all $n \ge N$. If this were not true, then
 $\exists \{m(k)\}_{k \ge 1} \subseteq \{n \ge N\}$ 
 such that $A(x_{m(k)}) \subseteq [f < \alpha + \frac{\epsilon}{2}]$
 for all $k$. It follows that
$H(x_{m(k)}) = \overline{co}(A(x_{m(k)})) \subseteq 
 [f \le \alpha + \frac{\epsilon}{2}]$ for each $k \ge 1$. 
 Since $u_{n(k)} \to u$, $\exists N_{1}$ such that for all $n(k) \ge N_1$, 
 $f(u_{n(k)}) \ge \alpha + \frac{3 \epsilon}{4}$. This leads to a contradiction.
\end{proof}
 In the following section, we show that (\ref{eq:sre}) is stable, provided $(A1)$-$(A3)$, $(S1)$ and $(S2)$ are satisfied. Following this, in Section~\ref{convergence}, we show that (\ref{eq:sre}) tracks a solution to a limiting DI, defined in terms of the ergodic occupation measures.
\section{Stability analysis} \label{stabilitytheorem}
For our analysis, we need to define a continuous trajectory, $\overline{x}([0, \infty))$, such that the limit of this trajectory coincides with that of $\{x_n\}_{n \ge 0}$. For this, we divide time axis as follows:
$t(0) := 0$ and $t(n) \ := \ \sum_{i=0}^{n-1} a(i)$, $\forall n \ge 1$. Then we let,
$\overline{x}(t(n)) \ := \ x_{n}$ $\forall n \ge 0$, and for 
$t \ \in \ (t(n), t(n+1))$ let {\footnotesize
\begin{equation*}
\overline{x}(t) \  :=  \ 
 \left( \frac{t(n+1) - t}{t(n+1) - t(n)} \right)\ \overline{x}(t(n))  +  
 \left(  \frac{t - t(n)}{t(n+1) - t(n)} \right) \  
\overline{x}(t(n+1)). 
\end{equation*}}

The time axis is further divided into intervals of approximate length $T$, where $T := T(\delta_2 - \delta_1) + 1$ and $\delta_1, \delta_2$ are defined in Section~\ref{assumptions}. Now, we define $T(\epsilon)$ for $\epsilon > 0$ as follows:
given a solution $x(\cdotp)$ to $\dot{x}(t) \in H(x(t))$ with $\lVert x(0) \rVert \le 1$, 
$x(t) \in N^\epsilon(\mathcal{A})$ for $t \ge T(\epsilon)$. Given $\epsilon > 0$, $\exists\ T(\epsilon)$ with the aforementioned property, since $\overline{B}_\epsilon(0)$ is a fundamental neighborhood of $\mathcal{A}$.
Define, $T_{0} \ := \ 0$ and $T_{n} \ := \ min\{ t(m)  :  t(m) \ge T_{n-1}+T\}$ for $n \ge 1$.
\textit{In other words, $\exists \{m(n)\}_{n \ge 0} \subseteq \mathbb{N}$
such that $T_{n} = t(m(n))$ for all $n \ge 0$}.

For the purpose of analyzing stability, we need to define the following rescaled trajectory:
$\hat{x}(t) := \frac{\overline{x}(t)}{r(n)}$, where $t \in [T_{n}, T_{n+1})$ and
$r(n) = \lVert \overline{x}(T_{n}) \rVert \vee 1$.
Also, let $\hat{x}(T_{n+1}^{-})\ := $  
$\underset{t \uparrow T_{n+1}}{\lim} \hat{x}(t)$.
We also define the rescaled martingale difference terms as follows:
$\hat{M}_{k+1} := \frac{M_{k+1}}{r(n)}$, $t(k) \in [T_{n}, T_{n+1})$.

Finally, we define the following piece-wise constant trajectories. Define $\hat{z}(t) := h_{r(n)}(\hat{x}(t(m)), y_m)$, where
$t \in [t(m), t(m+1))$ and $T_n \le t(m) < t(m+1) \le T_{n+1}$; also define $\overline{y}(t) := y_n$ for $t \in [t(n), t(n+1))$ and $n \ge 0$.

It can be readily verified that:
$E \left[ \lVert \hat{M}_{k+1} \rVert ^{2} | \mathcal{F}_{k} \right]$
$\le \ K \left( 1 + \lVert \hat{x}(t(k)) \rVert ^{2} \right)$.
The following lemma states that the Martingale noise convergences a.s. A proof can be found in  Borkar \& Meyn \cite{Borkar99}.
%
\begin{lemma}
\label{noiseconvergence}
$\underset{t \ge 0}{\sup}$ $E \lVert \hat{x}(t) \rVert ^{2}$ $< \infty$. Further,
the sequence $\hat{\zeta}_n$, $n \ge 0$, converges almost surely, 
where $\hat{\zeta}_{n}$ $:= \sum_{k=0}^{n-1} a(k) \hat{M}_{k+1}$ for all $n \ge 1$.
\end{lemma}
Let $x^n ([0,T])$ denote a solution up to time $T$
for $\dot{x}^{n}(\cdotp) = \hat{z}(T_{n} + \cdotp)$ with
$x^{n}(0) = \hat{x}(T_{n})$. Then,
$
 x^{n}(t)\ = \ \hat{x}(T_{n}) + \int_{0}^{t} \hat{z}(T_{n}+s) \,ds .
$
The following lemma states that the rescaled trajectories track the above defined solution trajectories.

\begin{lemma}
\label{difftozero}
 $\underset{n \to \infty}{\lim}$ 
 $\underset{t \in [T_{n}, T_{n}+T]}{\sup} \lVert
 x^n (t) - \hat{x}(t) \rVert = 0$ $a.s.$
\end{lemma}
\begin{proof} 
  Let $t \in \left[ t(m(n) + k), t(m(n) + k+1) \right)$ such that $T_n \le t(m(n)+k) < t(m(n)+k+1) \le T_{n+1}$, where $n \ge 0$.
  First we prove the lemma for $t(m(n)+k+1) < T_{n+1}$. Consider the following:
  \begin{eqnarray*}
   \hat{x}(t) & = & \left( \frac{t(m(n)+k+1) - t}{a(m(n)+k)} \right) \hat{x}(t(m(n)+k))
    \\ & + &
   \left( \frac{t - t(m(n)+k)}{a(m(n)+k)} \right) \hat{x}(t(m(n)+k+1)).
  \end{eqnarray*}
 Substituting for $\hat{x}(t(m(n)+k+1))$ in the above equation we get:
 \begin{multline}\nonumber
 \hat{x}(t)  =  \left( \frac{t(m(n)+k+1) - t}{a(m(n)+k)} \right) \hat{x}(t(m(n)+k))
   \\  +  \left( \frac{t - t(m(n)+k)}{a(m(n)+k)} \right)
   ( \hat{x}(t(m(n)+k)) + a(m(n)+k) \\ (
   h_{r(n)}(\hat{x}(t(m(n)+k)), y_{m(n)+k}) + \hat{M}_{m(n)+k+1} )),
 \end{multline}
 hence,
  \begin{multline}\nonumber
 \hat{x}(t) =  \hat{x}(t(m(n)+k)) + 
   \left( t - t(m(n)+k) \right) \\ \left( h_{r(n)}(\hat{x}(t(m(n)+k)), y_{m(n)+k}) + \hat{M}_{m(n)+k+1} \right).
  \end{multline}
 Unfolding $\hat{x}(t(m(n)+k))$ we get,
  \begin{multline}\label{eq:lemma41}
   \hat{x}(t) = \hat{x}(T_{n})  +  \sum_{l=0}^{k-1} a(m(n)+l)
  (  h_{r(n)}(\hat{x}(t(m(n)+l)), y_{m(n)+l}) \\ + \hat{M}_{m(n)+l+1}) +  
 \left(t-t(m(n)+k) \right) \\
  \left( h_{r(n)}(\hat{x}(t(m(n)+k)), y_{m(n)+k}) + \hat{M}_{m(n)+k+1}  \right). 
  \end{multline}
  Recall that
  \begin{equation}\nonumber
  x^{n}(t) = \hat{x}(T_{n}) + \int_0^t \hat{z}(T_{n}+s) \ \,ds .
  \end{equation}
  Splitting the above integral into several sub-integrals, we get
  \begin{equation}\nonumber
   x^n (t) =  \ \hat{x}(T_{n}) \ +\ \sum_{l=0}^{k-1}
  \int_{t(m(n)+l)}^{t(m(n)+l+1)} \hat{z}(s) \,ds
  + \int_{t(m(n)+k)}^{t} \hat{z}(s) \,ds .
  \end{equation}
  Thus,
 \begin{multline}\label{eq:lemma42}
  x^n(t) = \hat{x}(T_{n})  + \\ \sum_{l=0}^{k-1} a(m(n)+l)
 h_{r(n)}(\hat{x}(t(m(n)+l)), y_{m(n)+l}) + \\
 \left(t-t(m(n)+k) \right) h_{r(n)}(\hat{x}(t(m(n)+k)), y_{m(n)+k}). 
 \end{multline}
 From (\ref{eq:lemma41}) and (\ref{eq:lemma42}), 
  we get
  \begin{multline}\nonumber
  \lVert x^{n}(t) - \hat{x}(t) \rVert \le 
   \left\lVert \sum_{l=0}^{k-1} a(m(n)+l) \hat{M}_{m(n)+l+1} \right\rVert
  + \\ \left\lVert \left(t-t(m(n)+k) \right) \hat{M}_{m(n)+k+1} \ \right\rVert \text{ and}
  \end{multline}
  \begin{equation}\nonumber
 \lVert x^n (t) - \hat{x}(t) \rVert \le \lVert 
 \hat{\zeta}_{m(n)+k} - \hat{\zeta}_{m(n)} \rVert  + 
 \lVert \hat{\zeta}_{m(n)+k+1} - \hat{\zeta}_{m(n)+k} \rVert .
  \end{equation}
 If $t(m(n)+k+1) = T_{n+1}$ then replace $\hat{x}(t(m(n)+k+1))$ with
 $\hat{x}(T^-_{n+1})$ in the above set of equations and the arguments remain identical.
\end{proof}
Recall that $T = T(\delta_2 - \delta_1)+1$. We show that
$\{ x^n ([0, T])\ |\ n \ge 0 \}$ and $\{ \hat{x} ([T_n, T_n + T])\ |\ n \ge 0 \}$ are relatively compact
subsets of $C([0, T], \mathbb{R}^d)$, endowed with the sup-norm, $\lVert \cdotp \rVert _\infty$. To do this, we merely show that $\{ x^n ([0, T])\ |\ n \ge 0 \}$ is relatively compact. The relative compactness of $\{ \hat{x} ([T_n, T_n + T])\ |\ n \ge 0 \}$ follows from Lemma~\ref{difftozero}. In what follows, we show that $\{ x^n ([0, T])\ |\ n \ge 0 \}$ is equicontinuous and point-wise bounded. Relative compactness of $\{ x^n ([0, T])\ |\ n \ge 0 \}$ then follows from the \textit{Arzela-Ascoli Theorem}.

Recall that $\underset{t \ge 0}{\sup} \ E \lVert \hat{x}(t) \rVert ^2 < \infty$ \textit{a.s.} and
$\lVert \hat{z}(t) \rVert \le K(1 + \lVert \hat{x}([t])\rVert) $, where
$[t]:= \max\{t(m)\ |\ t(m) \le t \}$. Hence
$\underset{t \ge 0}{\sup} \ \lVert \hat{z}(t) \rVert < \infty$ \textit{a.s.}
For $\delta > 0$, we have:
\[
\lVert x^n (t + \delta) - x^n(t) \rVert \le \int _t ^{t+ \delta} \lVert \hat{z}(s) \rVert \,ds \le M \delta,
\]
where $M$ is a, possibly sample path dependent, constant such that
$\underset{t \ge 0}{\sup} \ \lVert \hat{z}(t) \rVert \le M$. In other words,
$\{ x^n ([0, T])\ |\ n \ge 0 \}$ is equicontinuous. Now, we need to show that the point-wise boundedness property is satisfied.
Since $\lVert x^n (0) \rVert = \lVert \hat{x}(T_n) \rVert \le 1$, it follows from Gronwall's inequality
that $\underset{n \ge 0}{\sup} \ \lVert x^n \rVert_\infty < \infty$ a.s., where
$\lVert x^n \rVert_\infty = \underset{t \in [0, T]}{\sup} \lVert x^n(t) \rVert$. In other words, $\{ x^n ([0, T])\ |\ n \ge 0 \}$ is point-wise bounded.

We are interested in showing that $\sup \limits_{n \ge 0} \lVert x_n \rVert < \infty$ a.s. We present a proof by contradiction. Suppose the iterates are unstable, then with positive probability we have
 $\underset{n \ge 0}{\sup} \ r(n) = \infty$ and there exists
$\{l\} \subseteq \{n\}$ such
that $r(l) \uparrow \infty$. 
Note that this sub-sequence may be sample path dependent.
The following lemma shows that 
$\{ \hat{x}([T_l, T_l +T]) \ |\ \{l\} \subseteq \{n\}\ \&\ r(l) \uparrow \infty\}$ tracks a solution to $\dot{x}(t) \in H(x(t))$ as $r(l) \uparrow \infty$.
\begin{lemma} \label{limitchar}
 Let $\{l\} \subseteq \{n\}$ such that $r(l) \uparrow \infty$. Any limit of \\
 $\{\hat{x}([T_l, T_l +T]) \ |\ \{l\} \subseteq \{n\}\ \&\ r(l) \uparrow \infty \}$ in 
 $C([0, T], \mathbb{R}^d)$ is of the form $x(t) = x(0) + \int_{0}^t\ z(s) \,ds$, where
 $x(0) \in \overline{B}_1(0)$ and $z: [0,T] \to \mathbb{R}^d$ is a measurable function
 such that $z(t) \in H(x(t))$, $t \in [0,T]$.
\end{lemma}
\begin{proof}
 Define the notation $[t] := \max \{t(m)\ |\ t(m) \le t \}$. For a fixed $n \ge 0$ 
 and $t_0 \in [T_n, T_{n+1})$, we have
 $\hat{z}(t_0) = h_{r(n)}(\hat{x}([t_0]), \overline{y}([t_0]))$ and 
 $\lVert \hat{z}(t_0) \rVert \le K(1+ \lVert \hat{x}([t_0]) \rVert)$. It follows from
 Lemma~\ref{noiseconvergence} that $\lVert \hat{z}(t) \rVert < \infty$ \textit{a.s.} Recall that
 $\{ \hat{x}(T_l + \cdotp)\ |\ \{l\} \subseteq \{n\} \}$ is relatively compact in $C([0, T], \mathbb{R}^d)$.
 Without loss of generality we may assume that: \textbf{(a)}
 $\hat{x}(T_l + \cdotp) \to x(\cdotp) \text{ in } C([0, T], \mathbb{R}^d)$, for some
  $x(\cdotp) \in C([0, T], \mathbb{R}^d)$; \textbf{(b)}
  $\hat{z}(T_l + \cdotp) \to z(\cdotp) \text{ weakly in } L^2([0, T], \mathbb{R}^d)$, for some 
  $z(\cdotp) \in L^2([0, T], \mathbb{R}^d)$.

From Lemma~\ref{difftozero} we get $x^l(\cdotp) \to x(\cdotp)$ in $C([0, T], \mathbb{R}^d)$.
Letting $r(l) \to \infty$ in
\begin{align*}
 x^l(t) &= x^l(0) + \int_0 ^t \hat{z}(T_l + s) \,ds, \text{ we get }
 \\
 x(t) &= x(0) = \int_0 ^t z(s) ds.
\end{align*}
Since $\lVert x^l(0) \rVert = \lVert \hat{x}(T_l) \rVert \le 1$ $\forall l$, we get 
$\lVert x(0) \rVert \le 1$. $\hat{z}(T_l + \cdotp) \to z(\cdotp)$ weakly in $L^2([0,T], \mathbb{R}^d)$, 
hence it follows from \textit{Banach-Saks Theorem} that
\begin{multline}\nonumber \exists \ \{k(l)\} \subseteq \{l\} \text{ such that }
 \frac{1}{N} \sum_{l=1}^{N} \hat{z}(T_{k(l)} + \cdotp) \to z(\cdotp) \\ \text{ strongly in }
 L^2([0,T], \mathbb{R}^d).
\end{multline}
Further,
\begin{multline} \label{limitchar1}
\exists \ \{m(N)\} \subseteq \{N\} \text{ such that } \\
 \frac{1}{m(N)} \sum_{l=1}^{m(N)} \hat{z}(T_{k(l)} + \cdotp) \to z(\cdotp) \text{ $a.e.$ on $[0,T]$}.
\end{multline}
Fix $t_0 \in [0,T]$ such that (\ref{limitchar1}) holds, \textit{i.e.},
\begin{align} \label{limitchar2}
\underset{m(N) \to \infty}{\lim} \ \frac{1}{m(N)} \sum_{l=1}^{m(N)} \hat{z}(T_{k(l)} + t_0) = z(t_0).
\end{align}
We know that $\hat{z}(T_{k(l)} + t_0) = h_{r(k(l))} 
(\hat{x}([T_{k(l)} + t_0]), \overline{y}([T_{k(l)} + t_0]))$. Note that
$\overline{y}([T_{k(l)} + t_0]) = \overline{y}(T_{k(l)} + t_0)$. 

We claim the following:
For any $\epsilon > 0$ there exists $N$ such that for all $n \ge N$
$\lVert \hat{x}(t(m)) - \hat{x}(t(m+1)) \rVert < \epsilon$, where $T_n \le t(m)< t(m+1) < T_{n+1}$.
If $t(m+1) = T_{n+1}$ then we claim that $\lVert \hat{x}(t(m)) - \hat{x}(T_{n+1}^-) \rVert < \epsilon$.
We shall prove this later, for now we assume it is true and proceed.

Since $\hat{x}(T_{k(l)} + t_0) \to x(t_0)$ it follows from the above claim that \\
$\hat{x}([T_{k(l)} + t_0]) \to x(t_0)$. Since $r(k(l)) \uparrow \infty$ it follows from
Lemma~\ref{hcross} that 
\begin{align*}
& \underset{r(k(l)) \uparrow \infty}{\lim} d \left(h_{r(k(l))} (\hat{x}([T_{k(l)} + t_0]), \overline{y}([T_{k(l)} + t_0])), H(x(t_0)) \right) 
\\ & = 0
 \text{  i.e.,} \ \underset{r(k(l)) \uparrow \infty}{\lim} d\left(\hat{z}(T_{k(l)} + t_0), H(x(t_0))\right) = 0.
\end{align*}
Further, since
$H(x(t_0))$ is convex and compact, it follows from equation (\ref{limitchar2}) that
$z(t_0) \in H(x(t_0))$. On the measure zero subset of [0,T] where (\ref{limitchar1}) does not hold, 
the value of $z(\cdotp)$ can be modified to ensure that $z(t) \in H(x(t))$ for all $t \in [0,T]$.

It is left to prove the claim that was made earlier. We first show that given
any $\epsilon > 0$ there exists $N$ such that $n \ge N$ implies that
$\lVert \hat{x}(t(m)) - \hat{x}(t(m+1)) \rVert < \epsilon$, where $T_n \le t(m)< t(m+1) < T_{n+1}$.
We know that 

{\footnotesize
\[ \hat{x}(t(m+1)) = \hat{x}(t(m)) + a(n) \left( h_{r(n)}(\hat{x}(t(m)), \overline{y}(t(m)))
 + \hat{M}_{n+1 } \right).
\]}

Hence,{\footnotesize
\begin{equation}\nonumber 
\lVert \hat{x}(t(m)) - \hat{x}(t(m+1)) \rVert \le
 a(n) \lVert h_{r(n)}(\hat{x}(t(m)), \overline{y}(t(m))) \rVert + \lVert \zeta_{n+1} - \zeta_{n} \rVert.
\end{equation}}

From (\ref{Hpwb}), the above inequality becomes{\footnotesize
\[\lVert \hat{x}(t(m)) - \hat{x}(t(m+1)) \rVert \le 
 a(n) K(1+ \lVert \hat{x}(t(m)) \rVert) + \lVert \zeta_{n+1} - \zeta_{n} \rVert.
\]}
It follows from Lemma~\ref{noiseconvergence} that
$a(n) K(1+ \lVert \hat{x}(t(m)) \rVert) \to 0$ and\\ $\lVert \zeta_{n+1} - \zeta_{n} \rVert \to 0$
respectively in the `almost sure' sense. In other words, there exists $N$ (possibly sample path dependent)
such that the claim holds. 
The second part of the unproven claim considers the situation when $t(m+1) = T_{n+1}$,
the proof of which follows in a similar manner.
\end{proof}
The following is an immediate corollary to the above lemma.
 \begin{corollary} \label{R_0}
  $\exists \ 1 < R_0 < \infty$ 
such that $\forall \ r(l) > R_0$
$\lVert \hat{x}(T_l + \cdotp) - x(\cdotp) \rVert < \delta_3 - \delta_2$,
where $\{ l \} \subseteq \mathbb{N}$ and $x(\cdotp)$ is a solution (up to time $T$) of $\dot{x}(t) \in H(x(t))$
such that $\lVert x(0) \rVert \le 1$. The form of $x(\cdotp)$ is given by
Lemma~\ref{limitchar}.
 \end{corollary}
\begin{proof}
Assume to the contrary that $\exists \ r(l) \uparrow \infty$ such that
$\hat{x}(T_l + \cdotp)$ is at least $\delta_3 - \delta_2$ away from any solution
to the $DI$. It follows from Lemma~\ref{limitchar} that
there exists a subsequence of 
$\{ \hat{x}(T_l + t), 0 \le t \le T \ :\ l \subseteq \mathbb{N} \}$
guaranteed to converge, in $C([0,T], \mathbb{R}^d)$,
to a solution of $\dot{x}(t) \in H (x(t))$ such that $\lVert x(0) \rVert \le 1$. 
This is a contradiction.
 \end{proof}
 \textit{It is worth noting that $R_0$ may be sample path dependent.
Since $T = T(\delta_2 - \delta_1) +1$ we get $\lVert \hat{x}([T_l + T]) \rVert < \delta_3$
for all $T_l$ such that $\lVert \overline{x}(T_l)\rVert (=r(l)) > R_0$.} We are now ready to state our stability theorem.
\begin{theorem}[The Stability Theorem] \label{stability}
 Under assumptions $(A1)-(A3), \ (S1) \ \& \ (S2)$, $\underset{n \ge 0}{\sup}\ \lVert x_n \rVert < \infty$ \textit{a.s.}
\end{theorem}
\begin{proof}
 Define $\mathcal{B} := \{\hat{\zeta}_n \text{ converges}\}$. It is enough to show that
 $\underset{n \ge 0}{\sup}\ \lVert x_n \rVert < \infty$ on $\mathcal{B}$. 
Suppose not, $\exists$ $\mathcal{D} \subseteq \mathcal{B}$ such that $P(\mathcal{D}) > 0$
and $\underset{n \ge 0}{\sup} \lVert x_n \rVert = \infty$ on $\mathcal{D}$. 
In other words, $\exists \{l\} \subseteq \{n\}$ such that $r(l) \uparrow \infty$.
Recall the notation $[t]=\max \{t(k) \mid t(k) \le t\}$ and the sequence $\{m(l)\}_{l \ge 0}$ such that 
$T_l = t(m(l))$.
There exists $N > 0$
such that the following hold simultaneously.
\begin{enumerate}
\item If $m(l) \ge N$ then $r(l) > R_0$ (see Corollary~\ref{R_0}).
 \item If $m(l) \ge N$ then $\lVert \hat{x}(T_l + T) \rVert < \delta_3$ (since $T=T(\delta_2 - \delta_1) +1$
 and from Corollary~\ref{R_0}).
 \item If $n>m \ge N$ then $\lVert \hat{\zeta}_n - \hat{\zeta}_m \rVert \le \epsilon_m$ ($\epsilon_m \to 0$ as $N \to \infty$).
 \item If $n \ge N$ then $a(n) < \frac{\delta_4 - \delta_3}{[K(1+C)+\epsilon_m]}$. Here $C$ is a sample path
 dependent constant such that $\underset{t \ge 0}{\sup} \lVert \hat{x}(t) \rVert \le C$.
\end{enumerate}
Given $T_{l}$ there exists $k$ such that $T_{l+1} = t(m(l)+k+1)$. 
Consider the following:
\begin{multline}\nonumber
 \hat{x}(T_{l+1}^-) = \hat{x}(t(m(l)+k)) + a(m(l)+k) \\ \left( h_{r(l)}(\hat{x}(t(m(l) + k)), y_{m(l)+k}) + \hat{M}_{m(l)+k+1} \right).
\end{multline}
Taking norms on both sides we get the following inequality,
\begin{multline}\nonumber
 \lVert \hat{x}(T_{l+1}^-) \rVert \le \lVert \hat{x}(t(m(l)+k)) \rVert + a(m(l)+k) \\ 
 \left( K(1 + \lVert \hat{x}(t(m(l) + k)) \rVert) + \lVert \hat{M}_{m(l)+k+1} \rVert \right).
\end{multline}
From the observations made earlier we get,
\begin{multline}\nonumber
 \lVert \hat{x}(T_{l+1}^-) \rVert \le \delta_3 + \frac{\delta_4 - \delta_3}{[K(1+C)+\epsilon_m]}
 \left( K(1+C) + \epsilon_m \right).
\end{multline}
Since $\lVert \overline{x}(T_l) \rVert = 1$ and $\lVert \hat{x}(T_{l+1}^-) \rVert < \delta_4$, we get
\begin{equation}
 \label{ratio}
 \frac{\lVert \overline{x}(T_{l+1}^-) \rVert}{\lVert \overline{x}(T_l)\rVert} = 
 \frac{\lVert \hat{x}(T_{l+1}^-) \rVert}{\lVert \hat{x}(T_l) \rVert} < \delta_4 < 1.
\end{equation}
Since $\lVert \overline{x}(T_l) \rVert \uparrow \infty$ and $\lVert \overline{x}(T_{l+1}^-) \rVert < \delta_4$.
It follows that the iterates have to make larger and larger jumps within a single interval of length $T$.
For all $m(l) > N$ the trajectory falls exponentially till it enters the ball
of radius $R_0$. This implies that $x(T_l) \in \overline{B}_{R_0}(0)$
when $r(l-1) < r(l)$. Further, the trajectory made a jump of at least $r(l) - R_0$ within time interval $T$.
This violates Gronwalls inequality.
\end{proof}
\section{Convergence analysis} \label{convergence}
In the previous section, we showed that the iterates given by (\ref{eq:sre}) are bounded almost surely (stable), provided the conditions in Section~\ref{assumptions} are satisfied. In this section, we show that our stability assumptions can be combined with additional assumptions, based on those found in Borkar \cite{BorkarM}, to provide a complete analysis of (\ref{eq:sre}). In other words, the main result, Theorem~\ref{convt}, of this section states the following: Recursion (\ref{eq:sre}) is bounded almost surely and converges to an internally chain transitive invariant set associated with a $DI$ that is defined in terms of the ergodic occupation measures associated with the `Markov process'. Below, we list the additional assumptions involved in our analysis.
 \\
 \textbf{(B1)} $\{y_n\}_{n \ge 0}$ is an $S$-valued Markov process with two associated 
 control processes: the iterate sequence
 $\{x_n\}_{n \ge 0}$ and another random process $\{z_n\}_{n \ge 0}$ taking values in a compact
 metric space $U$. Thus, for $n \ge 0$,
 \[
  P \left( y_{n+1} \in A \ |\ y_m, z_m, x_m, \ m \le n \right) = \int _A p(dy | y_n,z_n,x_n),
 \]
with $A$ Borel in $S$. The map
\[
 (y,z,x) \in S \times U \times \mathbb{R}^d \to p(dw | y,z,x) \in \mathcal{P}(S)
\]
is continuous, and further it is uniformly continuous in the $x$ variable on compacts with respect to the other
variables. $\mathcal{P}(S)$ is used to denote the space of probability measures on $S$.

Let $\varphi: S \to \varphi (S)$ with $\varphi (y, dz) \in \mathcal{P}(U)$ for each $y \in S$, be a measurable map. 
Suppose the Markov process has a (possibly non-unique) invariant probability measure
$\eta_{x, \varphi} (dy) \in \mathcal{P}(S)$,
we can define the corresponding \textit{ergodic occupation measure}
\begin{equation}\label{ergoc}
 \Psi_{x, \varphi} (dy, dz) := \eta_{x, \varphi}(dy) \varphi(dy, dz) \ \in \mathcal{P}(S \times U).
\end{equation}
Let $D(x)$ be the set of all such ergodic occupation measures for a prescribed $x$. It can
be shown that $D(x)$ is closed and convex for each $x \in \mathbb{R}^d$. 
Further, the map $x \mapsto D(x)$ is upper-semicontinuous.
For a proof of the aforementioned results
the reader is referred to \textit{Chapter 6.2} of \cite{BorkarBook}.
\\
\textbf{(B2)} $D(x)$ is compact.

Let us define a $\mathcal{P}(S \times U)$-valued random process $\mu (t) = \mu (t, dy,\ dz)$, $t \ge 0$, by
$
 \mu(t) := \delta_{y_n, z_n}, t \in [t(n), t(n+1)), \text{ for } n \ge 0.
$
For $t > s \ge 0$, define $\mu^t _s \in \mathcal{P}(S \times U \times [s,t])$ by
$\mu ^t _s (A \times B) := \frac{1}{t-s} \int_{B} \mu(y, A) \,dy$
for $A, B$ Borel in $S \times U , \ [s,t]$ respectively.
\\
\textbf{(B3)} Almost surely, for $t > 0$, the set $\{ \mu_s ^{s+t}, s \ge 0\}$ remains tight.

Define $\widetilde{h}(x, \nu) := \int h(x, y) \nu(dy, U)$ for $\nu \in \mathcal{P}(S \times U)$.
We use this to define the following \textit{DI}.
 \begin{equation} \label{hhat}
  \dot{x}(t) \in \hat{h}(x(t)), \text{ where } \hat{h}(x) := \{ \widetilde{h}(x, \nu)\ |\ \nu \in D(x)\}.
 \end{equation}
 
We are now ready to state the main result of this paper.
 \begin{theorem}[Stability \& Convergence] \label{convt}
  Under assumptions $(A1)-(A3)$, $(S1)$, $(S2)$ and $(B1)-(B3)$, almost surely the iterates given
  by (\ref{eq:sre}) are stable and converge to an internally chain transitive invariant set associated with
  $
   \dot{x}(t) \in \hat{h}(x(t)).
  $
 \end{theorem}
\begin{proof}
 Under assumptions $(A1)-(A3)$, $(S1)$ and $(S2)$, (\ref{eq:sre}) is stable as a consequence of
 Theorem~\ref{stability}.
 It now follows from \textit{Theorem 3.1} of Borkar \cite{BorkarM}
 that the iterates converge to an internally chain transitive invariant set associated with
  $
   \dot{x}(t) \in \hat{h}(x(t)).
  $
\end{proof}

\section{Applications}
In this section, we present two applications of the theory developed in Sections~\ref{assumptions} - \ref{convergence}. First, we analyze a generalized form of $TD(0)$, a simple yet effective TD algorithm. Second, we analyze an online-TD formulation for supervised learning with delayed feedbacks.

\subsection{Analysis of generalized $TD(0)$} \label{application}
Temporal difference $(TD)$ learning is an  reinforcement learning algorithm.
In this section, we consider a generalized version of online $TD(0)$, a simple yet effective TD algorithm.
Given a policy $\pi$, $TD(0)$ iteratively updates it's estimate $V$ of $V^{\pi}$, where
$V^{\pi}$ is the expected total reward associated with policy $\pi$.
When a non-terminal
state $s_n$ is visited at time $n$, the algorithm updates the estimate, $V(s_n)$, 
based on what happens `after the visit'.
The $TD(0)$ algorithm is given by,
\begin{equation}\label{td0}
 V(s_n) := V(s_n) + a(n) ( r_{n+1} + \gamma V(s_{n+1}) - V(s_n)).
\end{equation}
In the above, $s_n$ is the current state, $s_{n+1}$
is the next state, $r_{n+1}$ is the observed reward and $\gamma$ is the discount factor. 
Note that the starting state, $s_0$, is arbitrarily chosen. 
Define $\mathcal{F}_n := 
\sigma \left( s_0, \dots, s_n, r_1, \dots, r_n \right)$ for $n \ge 1$ and $\mathcal{F}_0 := \sigma \left( s_0 \right)$.
Also, define $M_{n+1} := \left[ r_{n+1} + \gamma V(s_{n+1}) \right] - E \left[ r_{n+1} + \gamma V(s_{n+1}) | 
\mathcal{F}_n \right]$ and $V^{\pi}(s_n) := E \left[ r_{n+1} + \gamma V(s_{n+1}) | 
\mathcal{F}_n \right]$ for $n \ge 0$.

In this section we assume that $s_n$ belongs to a compact state space, $S$. Further,
without loss of generality we may also assume that
\begin{equation}
 V^{\pi}(s_n) - V(s_n) = \Psi(s_n) V(s_n) + \psi(s_n),
\end{equation}
where $\Psi: S \to \mathbb{R}^{d \times d}$ and $\psi: S \to \mathbb{R}^d$.
In other words, $TD(0)$ given by (\ref{td0}) can be rewritten as:
\begin{equation}\label{td0r}
 V(s_n) := V(s_n) + a(n) ( \Psi(s_n) V(s_n) + \psi(s_n) + M_{n+1}).
\end{equation}
For a detailed exposition on temporal difference algorithms
the reader is referred to Tsitsiklis and Van Roy \cite{Tsitsiklis}.
In this section, we impose conditions on maps $\Psi$ and $\psi$ that guarantee the `stability 
and convergence' of the iterates given by $(\ref{td0})$.

\begin{remark}{
It is important to note that our $TD(0)$ algorithm (\textit{cf.} \ref{td0}) is more general than the regular 
$TD(0)$ update
with function approximation, as in (say) Tsitsiklis and Van Roy \cite{Tsitsiklis}. In particular,
the regular $TD(0)$ with function approximation can be written (see \cite{Tsitsiklis}) as in
(\ref{td0}). Note also that unlike the usual analyses of $TD(0)$, we do not assume that the Markov process
$\{s_n\}_{n \ge 0}$ \textbf{(a)} has a finite state and \textbf{(b)} is ergodic under the given stationary policy.}
\end{remark}
We state the first of our two assumptions below.
\\
 \textbf{(T1)} $\Psi: S \to \mathbb{R}^{d \times d}$ and
 $\psi: S \to \mathbb{R}^d$ are continuous maps.
 \\ \indent
  If we define a map $h: \mathbb{R}^d \times S \to \mathbb{R}^d$ as $h(v,s) \mapsto \Psi(s) v + \psi(s)$,
  then we have the following lemma.
  \begin{lemma}
   \label{lemma:A1sat}
   The map $h$ defined above is Lipschitz continuous in the first component. Further, the
   Lipschitz constant does not vary with the second component \textit{i.e.,} $h$ satisfies assumption $(A1)$.
  \end{lemma}
\begin{proof}
  Since $\Psi$ is continuous, it's range, $\Psi(S) \subset \mathbb{R}^{d \times d}$,
 is compact. It follows that
 \[ \lVert h(v_1, s) - h(v_2, s) \rVert \le \lVert \Psi(s) \rVert \times \lVert v_1 - v_2 \rVert
  \le L \lVert v_1 - v_2 \rVert,
 \]
 where $L := \underset{M \in A(S)}{\sup} \lVert M \rVert (< \infty)$.
Clearly, $L$ is independent of the second component.
\end{proof}
For the purpose of this section we assume that the martingale sequence defined earlier in this section, 
$\{M_n\}_{n \ge 0}$, satisfies assumption $(A2)$ and the step-size sequence, $\{ a(n)\}_{n \ge 0}$
satisfies assumption $(A3)$.
We are now ready to define the rescaled family of functions.
For $c \ge 1$, define $h_c(v, s) := \Psi(s)v + \psi(s)/c$, then the upper-limit of 
$\{h_c\}_{c \ge 1}$ is given by $h_{\infty}(v, s) = \{\Psi(s)v\}$. 
\begin{lemma}
 \label{S1sat}
 (\ref{td0}) satisfies assumption $(S1)$.
\end{lemma}
\begin{proof}
Let $c_n \uparrow \infty$, $s_n \to s$ and $\underset{n \to \infty}{\lim}$
$h_{c_n}(v, s_n) = u$. We need to show that $u = h_\infty (v, s)$. Since $\Psi$ is continuous,
$\underset{n \to \infty}{\lim} \Psi(s_n)v = \Psi(s)v$. Since $\psi$ is a bounded function,
$\underset{n \to \infty}{\lim} \psi(s_n)/ c_n = 0$. Hence, we get
$u = \underset{n \to \infty}{\lim} h_{c_n}(v, s_n) = \Psi(s)v \in h_\infty (v, s)$. 
\end{proof}
The following technical result is needed to state our second and final assumption.
\begin{lemma} \label{fundn}
 Let $H: \mathbb{R}^d \to \{ \text{subsets of }\mathbb{R}^d \}$ be a Marchaud map. 
 Let $\mathcal{A}$ be an associated attracting set that is also Lyapunov stable. Let
 $\mathcal{B}$ be a compact subset of
 the basin of attraction of $\mathcal{A}$. Then for all 
 $\epsilon > 0$ there exists $T(\epsilon)$ such that $\Phi_{[T(\epsilon), \infty)}(\mathcal{B})
 \subseteq N^{\epsilon}(\mathcal{A})$.
\end{lemma}
\begin{proof}
 Since $\mathcal{A}$ is Lyapunov stable, corresponding to $N^\epsilon (\mathcal{A})$ there exists 
 $N^\delta (\mathcal{A})$
 such that $\Phi_{[0, + \infty)} (N^\delta (\mathcal{A})) \subseteq N^\epsilon (\mathcal{A})$.
 Fix $x_{0} \in \mathcal{B}$. Since $\mathcal{B}$ is contained in the
 basin of attraction of $\mathcal{A}$, $\exists t(x_0) > 0$ such that 
 $\Phi_{t(x_0)}(x_0) \subseteq N^{\delta / 4} (\mathcal{A})$. Further, 
 from the upper semi-continuity of flow it follows that, for all $x \in N^{\delta(x_0)}(x_0)$, 
 $\Phi_{t(x_0)}(x) \subseteq N^{\delta / 4} (\Phi_{t(x_0)}(x_0))$ 
 for some $\delta(x_0) > 0$, see \textit{Chapter 2} of Aubin and Cellina \cite{Aubin}. 
 Hence $\Phi_{t(x_0)}(x) \subseteq N^\delta (\mathcal{A})$ for all $x \in N^{\delta(x_0)}(x_0)$. Since 
 $\mathcal{A}$ is Lyapunov stable, we get  $\Phi_{(t(x_0), + \infty]} (x) \subseteq N^\epsilon (A)$. In this manner for each $x \in \mathcal{B}$
 we calculate $t(x)$ and $\delta(x)$, the collection
 $\left\{N^{\delta(x)}(x) : x \in \mathcal{B} \right\}$
 is an open cover for $\mathcal{B}$. Let
 $\left\{N^{\delta(x_{i})}(x_{i}) \ |\  1 \le i \le m \right\}$ be a finite sub-cover. If
 we define
 $T(\epsilon) := \max \{ t(x_{i}) \ |\  1 \le i \le m \}$ then
 $\Phi_{[T(\epsilon), + \infty)} (\mathcal{B}) \subseteq N^\epsilon(\mathcal{A})$.
\end{proof}
If we define $H(v) := \overline{co} \left( 
\{ \Psi(s)v \ |\ s \in S \}\right)$ for all $v \in \mathbb{R}^d$, 
then it follows from Lemma~\ref{Hismarchaud} that $H$
is Marchaud. We state our second assumption below.
\\
\textbf{(T2)} Let $\epsilon > 0$ and 
$\mathcal{V}: \overline{B}_{1+ \epsilon} (0) \to [0, \infty)$. Let $\Lambda$
be a compact subset of $B_{1} (0)$, clearly
$\underset{u \in \Lambda}{\sup} \ \lVert u \rVert < 1$. The following hold:
 \textbf{(i)} For all $t \ge 0$, $\Phi_{t}(B_{1+ \epsilon} (0)) \subseteq B_{1+ \epsilon} (0)$,
 where $\Phi_t(\cdotp)$ is a solution to the $DI$ $\dot{x}(t) \in H(x(t))$;
 \textbf{(ii)} $\mathcal{V}^{-1}(0) = \Lambda$;
 \textbf{(iii)} $\mathcal{V}$ is a continuous map. For all $x \in \overline{B}_{1+ \epsilon} (0) 
 \setminus \Lambda$, $y \in \Phi_t (x)$ and $t > 0$ we have $\mathcal{V}(y) < \mathcal{V}(x)$.
 
\textbf{Proposition 3.25 from Bena\"{i}m et. al. \cite{benaim05} :} Under $(T2)$,
\textit{$\Lambda$ is a Lyapunov stable attracting set, further 
there exists an attractor, $\mathcal{A}$, contained in $\Lambda$
whose basin of attraction contains $B_{1+\epsilon}(0)$. }
\begin{lemma}
 (\ref{td0}) satisfies assumption $(S2)$.
\end{lemma}
\begin{proof}
Since $\mathcal{A} \subset \Lambda$ and $B_{1+ \epsilon}(0)$ is contained in the basin of attraction
of $\mathcal{A}$, it follows that
$\overline{B}_{1}(0) \subset B_{1+\epsilon}(0)$ is contained in the basin of attraction of $\Lambda$.
Since $\overline{B}_1(0)$ is compact it follows from Lemma~\ref{fundn} that it is
contained in some fundamental neighborhood
of $\Lambda$.
In this section, the attracting set associated with 
$ \dot{x}(t) \in \overline{co} \left( 
\{ A(y)x(t) \ |\ y \in S \}\right)$ in $(S2)$ is $\Lambda$.
\end{proof}
The following theorem is immediate.
\begin{theorem} \label{tdtheorem}
Under assumptions $(A1)$-$(A3)$, $(T1)$, $(T2)$ \& $(B1)$-$(B3)$, almost surely 
the iterates given by (\ref{td0}) are stable and converge to an internally chain transitive
invariant set associated with $\dot{x}(t) \in \hat{h}(x(t))$ ($\hat{h}$ is defined in Section~\ref{convergence}).
\end{theorem}
Let us analyze the special case when $\Psi$ is a constant map \textit{i.e.,} $\Psi(s) = M$ 
for some fixed $M \in \mathbb{R}^{d \times d}$, $s \in S$.
The recursion given by (\ref{td0}) becomes:
\begin{equation} \label{tditM}
 V(s_n) := V(s_n) + a(n) \left[ M V(s_n) + \psi(s_n) + M_{n+1} \right],
\end{equation}
and like before, $\psi : S \to \mathbb{R}^d$ is a continuous map. Clearly,
(\ref{tditM}) satisfies $(T1)$, 
$(A1)-(A3)$ and $(S1)$. The rescaled family of functions are
$ h_c(v, s) = Mv + \psi(s)/c \text{ and } h_{\infty}(v, s) = Mv; $
$H(v) = \overline{co} \left( \underset{s \in S}{\cup} h_\infty(v, s) \right) = \{Mv\}$.
Hence, the \textit{DI} $\dot{v}(t) \in H(v(t))$ is really the \textit{o.d.e.}
$\dot{v}(t) = Mv$, here.

If we assume that all eigenvalues of $M$ have strictly negative real parts, then
origin is the unique globally asymptotic stable equilibrium point 
(a globally attracting set that is also Lyapunov stable) of $\dot{x}(t) = M x(t)$
(see (\textit{11.2.3}) of Borkar \cite{BorkarBook}).

We define the following: 
  $(a)$ $\epsilon := 1$;
 $(b)$ $\mathcal{V}(v): \overline{B}_{2}(0) \to [0, \infty)$ as $\mathcal{V}(v) := \lVert v \rVert$;
 $(c)$ $\Lambda := \{0\}$. 

Solving the \textit{o.d.e.} $\dot{x}(t) = Mx(t)$, we get $\Phi_t (x(0)) = e^{M t} x(0)$ for $t \ge 0$.
Since all the eigenvalues of $M$ have strictly negative real parts, 
we have $\lVert \Phi_t (x(0)) \rVert < \lVert x(0) \rVert$, where $t > 0$ and $x(0) \in \mathbb{R}^d \setminus \{0\}$.  
Hence $\Phi_{t}(B_2 (0)) \subseteq B_2 (0)$ \textit{((T2)(i) is satisfied)}. 
It follows from the definition of $\mathcal{V}$ (see $(b)$) that $\mathcal{V}^{-1}(0) = \Lambda$ \textit{((T2)(ii)  is satisfied)}. 
Finally,
fix $v_0 \in B_{2}(0) \setminus \{0\}$ and $t > 0$, we have
$\lVert \Phi_t (v_0) \rVert <  \lVert v_0 \rVert$, hence $\mathcal{V}(\Phi_t (v_0)) < \mathcal{V}(v_0)$ \textit{((T2)(iii) 
 is satisfied)}. Now, it follows from Theorem~\ref{tdtheorem} 
 that the iterates given by (\ref{tditM}) are `stable and convergent'.
\subsection{A quick note on our assumptions}\label{sec_note}
In this section, we compare our framework with that of \cite{Tsitsiklis}. Traditional analysis of TD by Tsitsiklis and Van Roy \cite{Tsitsiklis} requires the following assumptions, among others:
\textbf{ (i)} The state space is finite.
  \textbf{(ii)} The Markov chain, associated with state evolution, is ergodic.
  \textbf{(iii)} The second moment of the single stage reward function is bounded, \textit{i.e.,} $\mathbb{E} r_{n+1} ^2 < \infty$ for $n \ge 0$.
 These conditions are very restrictive in real-world applications.
 Further, in \cite{Tsitsiklis}, the cost-to-go function is assumed to have the following form:
 $
  J(i_0, \theta) \approx \sum \limits_{k=1}^K \theta(k) \phi_k(i_0),
 $
where $\{\phi_k\}_{1 \le k \le K}$ are the basis functions. 

Below we state the additional stability assumptions in \cite{Tsitsiklis}, following which we briefly discuss how our assumptions differ from \cite{Tsitsiklis}. Note that the state space $S$ in \cite{Tsitsiklis} is finite.
 $\\ \\$
  \textbf{TB1} There exists a function $f: S \mapsto \mathbb{R}^+$ satisfying the following requirements:
    for all $i_0$, $1 \le k \le K$ and $m \ge 0$, we have
   $
    \sum \limits_{\tau = 0}^{\infty} \left\lVert \mathbb{E} \left[\phi_k(i_{\tau}) \phi'_k(i_{\tau + m}) \right] - \mathbb{E} \left[\phi_k(i_{t}) \phi'_k(i_{t + m}) \right] \right\rVert \le f(i_0)
   $
   and
  $
    \sum \limits_{\tau = 0}^{\infty} \left\lVert \mathbb{E} \left[\phi_k(i_{\tau}) r_{\tau + m+1} \right] - \mathbb{E} \left[\phi_k(i_{t}) r_{t + m+1} \right] \right\rVert \le f(i_0)   ,
  $
where $r_{n+1}$ is the reward at time $n$.
  \\
  \textbf{TB2} Given any $q > 0$, there exists $\mu_q$ such that for all $i_0$ and $t$:
$
 \mathbb{E} \left[ f^q (i_t)  \mid i_0 \right] \le \mu^q f^q (i_0).
$

The DI from $(S2)$ of Section~\ref{assumptions} is given by,
$H(x) := \overline{co} \left( \underset{y \in S}{\cup} h_\infty(x, y) \right)$ and $h_\infty(x, y) =$
$Limsup_{c \to \infty} \{h_c(x, y)\}$. 

As stated in \cite{Tsitsiklis}, $TB1$ and $TB2$ can be verified when the state space is finite. Further, $TB1$ and $TB2$ imply stability of TD when combined with previously stated assumptions, such as bounded second moments of the basis and reward functions, ergodicity of the Markov chain, etc., see \cite{Tsitsiklis} for details. On the other hand, our stability assumptions, $(S1)$ and $(S2)$, do not distinguish between finite and infinite state spaces. If $S$ is a compact metric space or if the associated Markov chain is not ergodic, our framework is readily applicable as opposed to \cite{Tsitsiklis}.
\begin{remark}
 In Section~\ref{sec_onlineTD}, our theory is used to provide an analysis of \textit{TD for supervised learning with delayed feedbacks}. For the problem considered therein, the controlled Markov process evolves in a compact state space. In other words, the framework of \cite{Tsitsiklis} cannot be used to provide an analysis of this algorithm.
\end{remark}

\subsection{Stability of online TD for supervised learning} \label{sec_onlineTD}
In this section, we consider the following weather forecasting problem described in \textit{Chapter 11} of Spall~\cite{spall2005}. On every day of the week, we are interested in issuing a forecast for the weather on the following Saturday, using current and past weather conditions and meteorological indicators. Unlike traditional supervised learning, in a TD formulation, the predictor is trained using successive predictions. 

We consider the online implementation of TD for supervised learning with delayed feedback, presented in Chapter 11 of Spall \cite{spall2005}. For this purpose we may use a predictor $f(\theta, \cdotp)$, that is parameterized by $\theta$. In other words, $f$ uses a sequence of inputs $\{x_0, x_1, \ldots, x_N\}$ to predict some (outcome) $Z$. The predictor $f$ is trained to minimize the following expected mean-squared error:
\begin{equation}
 \label{mse}
 \frac{1}{2} \mathbb{E} \left[ Z - f(\theta, x_n) \right]^2.
\end{equation}
Within the context of our weather forecasting problem, $Z$ is Saturday's weather outcome. The $x_i s$ represent the current and past weather conditions, among others, that are used as input for prediction. We consider the following $TD(\lambda)$ approach to training the predictor $f$,  in an online manner, see \cite{spall2005} for details:

 {\footnotesize
\begin{equation}
 \label{td}
 \theta_{n +1} = \theta_{n} + a(n) \left[ f(\theta_n, x_{n+1}) - f(\theta_n, x_{n}) \right] \sum \limits_{i=0}^n \lambda ^{n - i} \left[ \nabla _{\theta} f(\theta, x_i) \right] _{\theta = \theta_n} ,
\end{equation}
}
where
 \textbf{(i)} $0 \le \lambda \le 1$ is a fading factor which gives lower importance to past observations.
 \textbf{(ii)} $\{a(n)\}_{n \ge 0}$ is the standard step-size sequence.

When $\lambda = 0$, (\ref{td}) becomes $TD(0)$. Note that $TD(0)$ does not utilize older predictions in the current training step. Rewriting (\ref{td}) with $\lambda = 0$, we get:
\begin{equation}
 \label{td0_original_6}
 \theta_{n +1} = \theta_{n} + a(n) \left[ f(\theta_n, x_{n+1}) - f(\theta_n, x_{n}) \right] \left[ \nabla _{\theta} f(\theta, x_n) \right] _{\theta = \theta_n}.
\end{equation}
Let us suppose that the predictor is a linear regression function, \textit{i.e.,} $f(\theta, x) = \theta^T x$. For this to be effective, it is imperative to embed the input variables in a higher dimensional feature space. This embedding or feature extraction is often done using deep neural networks. Let $\phi(\cdotp)$ be the given feature function, then we may rewrite (\ref{td0_original_6}) as:
\begin{equation}
 \label{td0_pi_6}
 \theta_{n +1} = \theta_{n} + a(n) \theta_n^T \left( \phi(x_{n+1}) -  \phi(x_{n}) \right) \phi(x_n).
\end{equation}
There are numerous results in literature which discuss the convergence of (\ref{td0_pi_6}). The reader is referred to \cite{Tsitsiklis} or \cite{BertsekasBook} for details. However, all these results prove convergence of the algorithm under strong assumptions. Further it is hard to ensure stability, especially since the state space is continuous. Below, we present a complete analysis using our theory. Note that we consider the state space to be continuous, unlike previous analyses.

First, let us define the following objective function, $h$:
\begin{equation}
 \label{obj_function}
 h(\theta, z_1, z_2) := \left(\theta^T z_1 \right) z_2,
\end{equation}
where $z_1 = \phi(y_0) - \phi(y_1)$ and $z_2 = \phi(y_2)$ for some $\phi(y_0), \phi(y_1)$ and $\phi(y_2)$
in $\mathcal{K}$. Let us also define $S := (\mathcal{K} - \mathcal{K}) \times \mathcal{K}$, where $\mathcal{K} - \mathcal{K} = \{x - y \mid x, y \in \mathcal{K}\}$. In other words, $h$ is a function with domain $\mathbb{R}^d \times S$ and range $\mathbb{R}^d$. Now, (\ref{td0_pi_6}) can be rewritten as:
\begin{equation}
 \label{td0_1_6}
 \theta_{n +1} = \theta_{n} + a(n) h(\theta_n, y_n),
\end{equation}
where $y_n = \left( (\phi(x_{n+1}) -  \phi(x_{n})), \phi(x_n) \right)$ and, so from (\ref{obj_function})
$h(\theta_n, y_n) = \theta_n^T \left( \phi(x_{n+1}) -  \phi(x_{n}) \right) \phi(x_n)$. \textit{Note that the $\{y_n\}_{n \ge 0}$ process is the controlled Markov process of $(A1)(ii)$ in Section~\ref{assumptions}}.
We make the following assumptions on (\ref{td0_pi_6}):

 \textbf{(A)} The feature extracted input variables $\phi(x_n)$ belong to $\mathcal{K}$, which is a compact subset of $\mathbb{R}^d$.
 
 \textbf{(B)} The differential inclusion $\dot{\theta}(t) \in \underset{y \in S}{\cup} h(\theta, y)$ has an attractor inside the unit ball centered at the origin, such that its fundamental neighborhood is the closed unit ball itself.
 
 If we consider the previously mentioned weather forecasting problem, the input vector consists of bounded quantities such as atmospheric pressure, temperature, etc. Hence assumption $(A)$ is satisfied here as with many problems. It is also worth observing that assumption $(B)$ is a recast of $(S2)$ for (\ref{td0_pi_6}).

In Section~\ref{assumptions}, we have presented easily verifiable sufficient conditions for the stability of general recursions such as (\ref{td0_pi_6})/(\ref{td0_1_6}). It follows from the above discussion that (\ref{td0_pi_6}) satisfies $(A1) - (A3)$. In other words, to ensure stability of (\ref{td0_pi_6}), it is sufficient to show that $(S1)$ and $(S2)$ are satisfied.

It follows from the above definition of $h$, that the rescaled family of functions $\{h_c \mid c \ge 1 \}$, are such that $h_c(\theta, y) = h(\theta, y)$ for $\theta \in \mathbb{R}^d,\ y \in S$ and $c \ge 1$. Hence $h_\infty = h$. Since $h$ is continuous in the $y$-coordinate, (\ref{td0_pi_6}) satisfies $(S1)$. Suppose one is able to show that (\ref{td0_pi_6}) also satisfies $(S2)$, then it would follow from Theorem~\ref{stability}, that (\ref{td0_pi_6}) is stable. As stated earlier, it may be observed that $(S2)$ and $(B)$ are equivalent. One way to show that (\ref{td0_pi_6}) satisfies $(B)/(S2)$, is to construct an associated Lyapunov function and use Proposition 3.25 of Bena\"{i}m, Hofbauer and Sorin \cite{benaim05}. Constructing a Lyapunov function is often problem dependent. In summary, to ensure stability of (\ref{td0_pi_6}), one can check that $\dot{\theta}(t) \in \underset{y \in S}{\cup} h(\theta, y)$ has an attractor inside the unit ball centered at the origin, such that its fundamental neighborhood is the unit ball itself.

Once stability is assured, we then proceed to prove convergence of the stochastic iterates in Section~\ref{convergence}. \textit{
The above presented analysis can further be readily extended to $TD(\lambda)$ implementations of supervised learning with non-linear predictors, \textit{i.e.,} for the case of $\lambda \neq 0$.}

\textbf{A note on assumption (B)}. Let us suppose that the state space in the above considered problem is finite. Further, let us suppose that the following discounted version of (\ref{td0_original}) is analyzed:
\begin{equation}
 \label{td0_original}
 \theta_{n +1} = \theta_{n} + a(n) \left[ \gamma f(\theta_n, x_{n+1}) - f(\theta_n, x_{n}) \right] \left[ \nabla _{\theta} f(\theta, x_n) \right] _{\theta = \theta_n},
\end{equation}
where $0 < \gamma < 1$ is the discount factor. Let the feature matrix be full rank (as in \cite{Tsitsiklis}) but the Markov chain be non-ergodic (unlike \cite{Tsitsiklis}). Then, the limiting DI is given by
\begin{equation}
\label{di_note_1}
 \dot{\theta}(t) \in \underset{D \in \bar{D}}{\cup} \Phi^T D (\gamma P - I)\Phi \theta(t),
\end{equation}
where $D$ is a diagonal matrix containing a given stationary distribution of the Markov chain on the diagonal. Let $\bar{D}$ denote the space of all diagonal matrices with stationary distribution as diagonal entries. Additionally, if $\forall D \in \bar{D}$ the diagonal entries are positive (the diagonal entries of any given $D$ sum to one), then $\Phi^T D (\gamma P - I)\Phi$ is a negative definite matrix for each $ D \in \bar{D}$, cf. proof of Theorem 1 in \cite{Tsitsiklis}. Hence, origin is the global asymptotic stable equilibrium point associated with the above DI, (\ref{di_note_1}). In other words, assumption (B) is satisfied. The weather forecasting problem, on the other hand, can be thought of as a shortest path problem with termination being on Saturday, and can be similarly dealt with since a shortest path problem under a proper policy can be cast as a discounted cost problem and vice versa, cf. \textit{Chapter 2.3} of \cite{BertsekasBook}.

\bibliographystyle{abbrv}
\bibliography{references}

\end{document}